\documentclass[12pt]{amsart}

\usepackage{amssymb,amsthm,amsmath}
\usepackage[numbers,sort&compress]{natbib}
\usepackage{color}
\usepackage{graphicx}
\usepackage{tikz}
\usepackage{amssymb,amsthm,amsmath}
\usepackage[numbers,sort&compress]{natbib}
\usepackage{color}
\usepackage{graphicx}
\usepackage{tikz}

\title[ ]{ Fermi isospectrality for discrete periodic Schr\"odinger operators}

\author{Wencai Liu}
\address[W. Liu]{ Department of Mathematics, Texas A\&M University, College Station, TX 77843-3368, USA} \email{liuwencai1226@gmail.com; wencail@tamu.edu}



\theoremstyle{plain}
\newtheorem{theorem}{Theorem}[section]

\newtheorem{corollary}[theorem]{Corollary}
\newtheorem{lemma}[theorem]{Lemma}

\newtheorem{remark}{Remark}
\newcommand{\C}{\mathbb{C}}

\newcommand{\Z}{\mathbb{Z}}

\newcommand{\R}{\mathbb{R}}

\theoremstyle{plain}
\newtheorem{definition}{Definition}
\newtheorem{conjecture}{Conjecture}

\begin{document}
	
	
	\begin{abstract}
	Let  $\Gamma=q_1\mathbb{Z}\oplus q_2 \mathbb{Z}\oplus\cdots\oplus q_d\mathbb{Z}$, where	$q_l\in \mathbb{Z}_+$, $l=1,2,\cdots,d$.
		Let $\Delta+V$    be the discrete Schr\"odinger operator, 
		where $\Delta$ is the discrete Laplacian on $\mathbb{Z}^d$ and the potential $V:\mathbb{Z}^d\to \mathbb{R}$ is $\Gamma$-periodic. 
		We prove  three   rigidity theorems for   discrete periodic Schr\"odinger operators in any dimension $d\geq 3$:
	\begin{enumerate}
				\item  	if at some energy level,    Fermi varieties of the $\Gamma$-periodic potential $V$  and the $\Gamma$-periodic potential $Y$  are the same (this feature is referred to as {\it Fermi isospectrality} of $V$ and $Y$), and  $Y $ is a separable   function, then $V$ is separable;
					\item  if   potentials  $V$ and $Y$ are { Fermi isospectral} and both  $V=\bigoplus_{j=1}^rV_j$ and $Y=\bigoplus_{j=1}^r  Y_j$ are separable   functions,  then, up to a constant,  lower dimensional decompositions  $V_j$  and $Y_j$ are Floquet isospectral, $j=1,2,\cdots,r$;
		\item 	if a  potential $V$ and    the zero potential are {Fermi isospectral}, then $V$ is zero.

	\end{enumerate}
In particular, all conclusions in (1), (2) and (3) hold if we replace the assumption  ``Fermi isospectrality" with a stronger assumption ``Floquet isospectrality".

	\end{abstract}
	
	\maketitle 
	\section{Introduction and main results}

	Given $q_l\in \Z_+$, $l=1,2,\cdots,d$,
	let $\Gamma=q_1\Z\oplus q_2 \Z\oplus\cdots\oplus q_d\Z$.
	We say that a function $V: \Z^d\to \R$ is $\Gamma$-periodic if 
	for any $\gamma\in \Gamma$, $V(n+\gamma)=V(n)$. 
We  assume that $q_l$, $l=1,2,\cdots, d$, are relatively prime.  We  regard the zero function  as a $\Gamma$-periodic potential. 
 In this paper, $\Gamma$ is fixed. For simplicity, we sometimes 
		call  a function $V$ periodic instead of $\Gamma$-periodic. 

	Let $\Delta$ be the discrete Laplacian on $\ell^2(\Z^d)$, namely
	\begin{equation*}
	(\Delta u)(n)=\sum_{||n^\prime-n||_1=1}u(n^\prime),
	\end{equation*}
	where $n=(n_1,n_2,\cdots,n_d)\in\Z^d$, $n^\prime=(n_1^\prime,n_2^\prime,\cdots,n_d^\prime)\in\Z^d$ and 
	\begin{equation*}
	||n^\prime-n||_1=\sum_{l=1}^d |n_l-n^\prime_l|.
	\end{equation*}
	Consider the discrete  Schr\"{o}dinger operator on $\ell^2({\Z}^d)$,
	\begin{equation} \label{h0}
	H_0=\Delta +V ,
	\end{equation}
	where $V$ is $\Gamma$-periodic.

	Let $\{\textbf{e}_j\}$, $j=1,2,\cdots d$, be the standard basis in $\Z^d$:
	\begin{equation*}
	\textbf{e}_1 =(1,0,\cdots,0),\textbf{e}_2 =(0,1,0,\cdots,0),\cdots, \textbf{e}_{d}=(0,0,\cdots,0,1).
	\end{equation*}
	Let us consider  the equation
		\begin{equation} 
		(\Delta u)(n)+V(n)u(n)=\lambda u(n) \label{spect_0}, n\in\Z^d,
	\end{equation}
with the so called Floquet-Bloch boundary condition
	\begin{equation}  
	u(n+q_j\textbf{e}_j)=e^{2\pi i k_j}u(n),j=1,2,\cdots,d, \text{ and } n\in \Z^d.\label{Fl}
	\end{equation}

In this paper, we are interested in the  inverse problem of \eqref{spect_0} and \eqref{Fl}. 
Let $D_{V} (k)$ be the periodic operator $\Delta+V$ with the Floquet-Bloch boundary condition \eqref{Fl} (see Section \ref{S2} for the precise description of $D_{V} (k)$).
Two $\Gamma$-periodic potentials $V$ and $Y$ are called 
Floquet isospectral if  
\begin{equation}\label{gfi}
\sigma(D_{V} (k))= \sigma(D_{Y} (k)), \text{ for any } k \in\R^d.
\end{equation}

We say that a function $V$ on $\Z^d$ is $(d_1,d_2,\cdots,d_r)$ separable (or simply separable), where  $\sum_{j=1}^r d_j= d$ with $r\geq 2$, if there exist functions 
$V_j$  on $\Z^{d_j}$, $j=1,2,\cdots,r$,   
such that  for any $(n_1,n_2,\cdots,n_d)\in\Z^d$,
\begin{align}
V(n_1,n_2,\cdots,n_d)=&V_1(n_1,\cdots, n_{d_1})+V_2(n_{d_1+1},n_{d_1+2},\cdots,n_{d_1+d_2})\nonumber\\
&+\cdots+V_r(n_{d_1+d_2+\cdots +d_{r-1}+1},\cdots,n_{d_1+d_2+\cdots +d_r}).\label{g61}
\end{align}

We   say  $V:\Z^d\to \R$ is completely  separable if $V$ is $(1,1,\cdots, 1)$ separable. 
When there is no ambiguity, we write down \eqref{g61} as $V=\bigoplus_{j=1}^r V_j$.

Our interest in this paper is first motivated by several questions asked   by  Eskin, Ralston and Trubowitz \cite{ERT84,ERTII}, and Gordon and Kappeler \cite{gki} : 
\begin{enumerate}
	\item [Q1.]   If  $Y$ and $V$ are Floquet isospectral,  and $Y$ is (completely) separable,  is $V$ (completely) separable?
		\item [Q2.] Assume that both $Y=\bigoplus_{j=1}^dY_j$ and $V=\bigoplus_{j=1}^d V_j$   are completely separable. If $V$ and $Y$ are Floquet   isospectral,   are the one-dimensional potentials $V_j$ and $Y_j$ Floquet  isospectral (up to possible translations)?
\end{enumerate}  
While Q1 and Q2 were   formulated in the continuous case  \cite{ERT84,ERTII,gki,ksurvey}, the  same questions  apply to  the discrete case. See more details in a recent survey article  by Kuchment \cite[Section 5]{ksurvey}.  
For both continuous and discrete periodic Schr\"odinger operators,  Q1 and Q2 have been partially answered   by  Eskin-Ralston-Trubowitz \cite{ERT84,ERTII}, Gordon-Kappeler \cite{gki} and Kappeler \cite{kapiii}.   See  more details in Remark \ref{re1}.

As corollaries of our main results, Q1 and Q2 will be completely  resolved  in the discrete case.  
Another motivation of this paper comes from  the work of    B{\"a}ttig,  Kn\"orrer and Trubowitz  \cite{bktcm91}, where the authors  established several rigidity results based on  the Fermi variety 
when $d=3$.  See item  (a) in Remark \ref{re1}.

 Before stating our results, let us introduce some terminology.

	\begin{definition}{\rm
	Given $\lambda\in \C$,
	the  { \it Fermi surface (variety) } $F_{\lambda}(V)$ consists of all $k=(k_1,k_2,\cdots,k_d)\in \C^{d}$ for which
	there exists a non-zero solution of  ~ \eqref{spect_0} and ~\eqref{Fl}.}
	
\end{definition}

The main goal of this paper is to  study    isospectral problems for given  Fermi varieties. 
\begin{definition}\label{fermiiso}
	{\rm 
	Let $V$ and $Y$ be two $\Gamma$-periodic functions. We say $V$ and $Y$ are  {\it Fermi isospectral} if 
	${F}_{\lambda_0} (V)={F}_{\lambda_0} (Y)$ for some $\lambda_0\in\C$. }
\end{definition}

Motivated by    Q1 and Q2, and the work of   B{\"a}ttig,  Kn\"orrer and Trubowitz  \cite{bktcm91}  we ask two   questions:
\begin{enumerate}
	\item [Q3.]   If  $Y$ and $V$ are  { Fermi isospectral},  and $Y$ is (completely) separable,  is $V$ (completely) separable?
	\item [Q4.] Assume that both $Y=\bigoplus_{j=1}^rY_j$ and $V=\bigoplus_{j=1}^r V_j$   are   separable. If $V$ and $Y$ are   { Fermi isospectral},   are the  lower dimensional potentials $V_j$ and $Y_j$ Floquet  isospectral?
\end{enumerate}  
Compared  to  Q2, question  Q4 allows  more general decompositions, namely, $V_j$ and $Y_j$ are not necessarily   one dimensional. 

 From Remark \ref{re3} in Section \ref{S2}, one can see that 
\begin{enumerate}
	\item  { Floquet isospectrality } of $V$ and $Y$  is equivalent to saying that eigenvalues of equation \eqref{spect_0} and \eqref{Fl} with potentials 
	$V $ and $Y$ are the same, and  { Fermi isospectrality} of $V$ and $Y$ means  products of eigenvalues are the same;
\item Floquet isospectrality of $V$ and $Y$ is equivalent to   $F_{\lambda}(V)=F_{\lambda}(Y)$ for any $\lambda\in\C$. In other words, Bloch varieties of $V$ and $Y$ are the same.  The {\it Bloch variety  }  is defined as follows:
$$B(V)=\{(k,\lambda)\in\C^{d+1}: k\in F_{\lambda} (V)\}.$$
  
\end{enumerate}
	Therefore, the Fermi isospectrality contains much less information than Floquet isospectrality.

 As readers can see above,  the Floquet isospectrality problem ~\eqref{gfi} can be reformulated in terms of 
the Fermi  variety  and Bloch variety. The reformulation allows us  to use  tools from algebraic geometry and complex analysis in multi-variables to investigate the equation ~\eqref{spect_0} with  the boundary condition ~\eqref{Fl}.
Fermi and Bloch varieties also play  significant roles in the study of the spectral theory of   periodic  Schr\"odinger operators,  such as 
the  existence of embedded eigenvalues   ~\cite{kv06cmp,kvcpde20,shi1,IM14,AIM16,liu1,dks}.

The main contribution of this paper is   {establishing}  several rigidity theorems of  Fermi isospectrality  for  discrete periodic Schr\"odinger operators. In particular, we  answer Q3 and Q4  affirmatively   for any dimension $d\geq 3$,  and thus answer  Q1 and Q2 as well.

\begin{theorem}\label{thmmain2}
	Let  $d\geq 3$. Assume that $V$ and $Y$ are Fermi isospectral, and $Y$ is $(d_1,d_2,\cdots,d_r)$ separable, then
	$V$ is $(d_1,d_2,\cdots,d_r)$ separable.
\end{theorem}

We say that two functions $G_1$ and $G_2$  are Floquet isospectral up to a constant if there exists a constant $C$ such that $G_1+C$ and $G_2$ are Floquet isospectral.
\begin{theorem}\label{thmmain3}
	Let $d\geq 3$.
	Assume  that potentials  $V$ and $Y$ are  Fermi isospectral,   and  both  $V=\bigoplus_{j=1}^rV_j$ and $Y=\bigoplus_{j=1}^r  Y_j$ are separable   functions.  Then, up to a constant,  lower dimensional decompositions  $V_j$  and $Y_j$ are Floquet isospectral, $j=1,2,\cdots,r$.
\end{theorem}
 
Theorems \ref{thmmain2} and \ref{thmmain3} imply  
\begin{theorem}\label{mainthm}
	Let  $d\geq 3$. Assume that $V$ and $Y$ are Fermi isospectral, and $Y=\bigoplus _{j=1}^r Y_j$ is    separable. Then $V=\bigoplus _{j=1}^r V_j$ is   separable.   Moreover, up to a constant,
	$Y_j$ and 	$V_j$, $j=1,2,\cdots,r$,   are Floquet isospectral.
\end{theorem}
As a corollary of Theorem \ref{mainthm}, we have
\begin{corollary}\label{coro2}
	Let  $d\geq 3$. Assume that $V$ and $Y$ are Fermi isospectral, and $Y=\bigoplus _{j=1}^d Y_j$ is completely  separable. Then $V=\bigoplus _{j=1}^d V_j$ is completely separable.   Moreover, up to a constant,
$Y_j$ and 	$V_j$, $j=1,2,\cdots,d$,   are Floquet isospectral.
\end{corollary}
Theorem \ref{mainthm} allows us to establish  two Ambarzumian-type  theorems.
We say that a  function  $V$ on $\Z^d$ only depends on the first $\tilde{d}$ variables with $\tilde{d}<d$ if there exists a function $\tilde{V}$ on $\Z^{\tilde{d}}$ such that  for any $(n_1,n_2,\cdots,n_d)\in\Z^d$, $$V(n_1,n_2,\cdots,n_{\tilde d},n_{\tilde {d}+1}\cdots,n_d)=\tilde{V}(n_1,n_2,\cdots,n_{\tilde{d}}).$$
\begin{theorem}\label{thmmain5}
	
	Let $d\geq 3$.   
	Assume   that $V$  and    $Y$  are Fermi isospectral, and  $Y$  only depends on the first $\tilde{d}$ variables with $\tilde{d}<d$.
	Then
	$V$  only depends on the first $\tilde{d}$ variables.
\end{theorem}
\begin{theorem}\label{thmmain}
	
	Let $d\geq 3$.  
	Assume   that $V$  and   the zero potential  are Fermi isospectral. 
	Then
	$V$ is zero.
\end{theorem}

\begin{remark}\label{re00}
For discrete periodic Schr\"odinger operators, Floquet isospectrality implies Fermi isospectrality (see Remark \ref{re3}),  so all results in this paper hold if we replace the assumption ``Fermi isospectrality  of $Y$ and $V$" with ``Floquet isospectrality of $Y$ and $V$".
\end{remark}
\begin{remark}\label{re0}
	Theorem \ref{thmmain}    immediately follows from  Theorem \ref{thmmain5}.
	We choose to present    a direct  proof to better illustrate one of the main ideas in this paper.
\end{remark}
 The study of    isospectral problems   of periodic Schr\"odinger operators is a fascinating and vast subject that can not be reviewed entirely here. 
We only list  a few  of the most related results: 
\begin{remark}\label{re1}
	 \begin{enumerate}
	 	\item[(a)]	For continuous  periodic Schr\"odinger operators,  B{\"a}ttig,  Kn\"orrer and Trubowitz   proved   Theorems   \ref{thmmain2}  and \ref{thmmain}   when $d=3$  \cite{bktcm91}. Their proof is based on the  directional compactification, which differs  from our approach. 
		\item[(b)] 
	For the discrete case, 
	 Kappeler proved that if   $V$ and $Y$ are Floquet isospectral,  and $Y$ is completely separable,  then $V$ is completely separable \cite{kapiii}.  Kappeler also answered Q2 affirmatively \cite{kapiii}. 
		\item[(c)] For the continuous case,  Eskin, Ralston and Trubowitz  \cite{ERT84,ERTII} proved that  if $V$ and $Y$ are   Floquet isospectral, and $Y$ is completely  separable, then
		$V$ is completely separable.   
		\item[(d)] For the continuous case,    Gordon  and Kappeler \cite{gki} proved that  if $V$ and $Y$ are   Floquet isospectral, and $Y$ is   separable, then
		$V$ is  separable.   
			\item[(e)]   For the continuous case in dimensions two and three,  Gordon and  Kappeler  \cite{gki} answered Q2 affirmatively.
	\end{enumerate}
\end{remark}
The statements in Remark \ref{re1}  hold under  some assumptions on  lattices and functions.  
We refer readers to
~\cite{ksurvey,ERT84,MT76,kapiii,Kapi,Kapii,ERTII,gki,wa,gkii,gui90,eskin89} and references therein for precise descriptions of those  assumptions  and    other related developments.

 Finally,
 we want to  present some ideas of the proof. 
It is known (see Section \ref{S2})  that in appropriate coordinates the Fermi variety is algebraic. Moreover, it is defined as the set of zeros of a Laurent polynomial  (see Lemma  \ref{le1}). 
 This is our starting point.  Except for   Lemma \ref{le1},   proofs in this paper  are entirely self-contained and 
the approaches  are  new. 
  
 Our strategy is to focus  on the study of the Laurent polynomial.
 One of the ingredients  in the proof is to understand the linear independence of three multiplicative   groups consisting of $a$th, $b$th and $c$th ($a$, $b$ and $c$ are relatively prime) roots of unity (see Lemma \ref{key2}). 
 Another ingredient is to find  appropriate algebraic curves  so that    the  leading (Laurent) polynomials can be   explicitly calculated   along them. 
 
 	The rest of this paper is organized as follows.  
 	 In Section \ref{S2}, we  recall some basics    concerning  the Fermi variety and the discrete Floquet-Bloch transform. 
 	In Section \ref{S3},  we provide several general technical lemmas, which are independent of   isospectral problems. 
 	In Section \ref{S4}, we    provide several  other technical results related to Fermi isospectral problems. 
 	Sections \ref{S5}, \ref{S6} and \ref{S7}  are  devoted to proving  Theorems \ref{thmmain2},  \ref{thmmain3}, and  \ref{thmmain5} 
 and	  \ref{thmmain}    respectively.

		\section{ Basics }\label{S2}
	
	In this section, we first recall some basic facts about the Fermi variety, see, e.g., \cite{liu1,ksurvey,liujmp22}.

	Let $\C^{\star}=\C\backslash \{0\}$ and $z=(z_1,z_2,\cdots,z_d)$. 
For any $z\in (\C^\star)^d$, consider the  equation
\begin{equation}\label{gei1}
 (\Delta u)(n)+V(n)u(n)=\lambda u (n),n\in\Z^d,
\end{equation} 
with boundary conditions
\begin{equation}\label{gei2}
u(n+q_j\textbf{e}_j)=z_j u(n), j=1,2,\cdots, d, \text{ and } n\in \Z^d.
\end{equation}

Now we 
introduce a fundamental domain $W$ for $\Gamma=q_1\Z\oplus q_2\Z\oplus \cdots \oplus q_d\Z$:
\begin{equation*}
W=\{n=(n_1,n_2,\cdots,n_d)\in\Z^d: 0\leq n_j\leq q_{j}-1, j=1,2,\cdots, d\}.
\end{equation*}
By writing out $H_0=\Delta +V$ as acting on the $Q=q_1q_2\cdots q_d$ dimensional space $\{u(n),n\in W\}$, the equation \eqref{gei1} with boundary conditions \eqref{gei2} 
translates into the eigenvalue problem for a $Q\times Q$ matrix $\mathcal{D}_V(z)$. 
We refer readers to   \cite[ Sections 2 and 3]{GKTBook} and  \cite[Section 2]{pete} 
	for  more detailed  descriptions  of  $\mathcal{D}_V(z)$. 

Let  
$\mathcal{P}_V(z,\lambda)$ be the determinant of 
$\mathcal{D}_V(z)-\lambda I$.  
Let $z_j=e^{2\pi i k_j}, j=1,2,\cdots, d$, $ D_V(k)=\mathcal{D}_V(z)$ and $P_V(k,\lambda)=\mathcal{P}_V(z,\lambda)$. 
 Therefore, the Fermi variety $F_{\lambda}(V)$ 
is determined by 
\begin{equation}\label{gfer}
F_{\lambda}(V)=\{k\in \C^d: {P}_{V} (k,\lambda ) =0 \}.
\end{equation}
For each $k\in\R^d$, it is easy to see that  $D_V(k)$ has $Q=q_1q_2\cdots q_d$ eigenvalues. Order them in non-decreasing  order
\begin{equation*}
\lambda^1_V(k)\leq \lambda^2_V(k)\leq\cdots \leq \lambda^Q_V(k).
\end{equation*}
Therefore,
\begin{equation}\label{g60}
P_V(k,\lambda)=\prod_{m=1}^Q(\lambda^m_V(k)-\lambda).
\end{equation}

One can see that 
$\mathcal{P}_V(z,\lambda) $ is a polynomial in the variables $\lambda$ and 
$z_1,z_1^{-1},\cdots, z_d,z_d^{-1},$    with highest degrees terms,
$\kappa_1z_1^{ \frac{Q}{q_1}}, \kappa_1 z_1^{- \frac{Q}{q_1}}, \kappa_2z_2^{ \frac{Q}{q_2}},\kappa_2z_2^{- \frac{Q}{q_2}}, \cdots ,\kappa_dz_d^{ \frac{Q}{q_d}},\kappa_d z_d^{- \frac{Q}{q_d}}$ and $(-1)^Q\lambda^{Q}$, where 
\begin{equation*}
\kappa_j=(-1)^{\frac{q_j-1}{q_j}Q}\in\{-1,1\}, j=1,2,\cdots,d.
\end{equation*}
In other words $\mathcal{P}_V(z,\lambda) $ is a Laurent polynomial of $\lambda$ and 
$z_1, z_2,\cdots, z_d$, and
\begin{equation}\label{g37}
\mathcal{P}_V(z,\lambda) = \kappa_1z_1^{ \frac{Q}{q_1}}+\kappa_1 z_1^{- \frac{Q}{q_1}}+\kappa_2z_2^{ \frac{Q}{q_2}}+\kappa_2z_2^{- \frac{Q}{q_2}}+ \cdots +\kappa_dz_d^{ \frac{Q}{q_d}}+ \kappa_dz_d^{- \frac{Q}{q_d}} +(-1)^Q\lambda^{Q}+\cdots.
\end{equation}

A  single term Laurent polynomial, i.e., $Cz_1^{a_1}z_2^{a_2}\cdots z_{d}^{a_d}$, where $a_j\in\Z$, $j=1,2,\cdots,d$,
and $C$ is a non-zero constant,    is called monomial. 
\begin{definition}
	We say that a Laurent polynomial $h\left(z_{1}, z_{2},\cdots,z_d\right)$ is irreducible if it can not be factorized non-trivially, that is, there are no non-monomial Laurent polynomials $f\left(z_{1}, z_{2},\cdots,z_d\right)$ and $g\left(z_{1}, z_{2},\cdots,z_d\right)$ such that 
	$h=fg$.
\end{definition}

\begin{theorem}\label{thm2}
	\cite{liu1}
	Let $d\geq 3$. 
	The Laurent polynomial $ \mathcal{P}_V(z,\lambda)$ (as a function of $z\in\C^d$) is irreducible for any $\lambda\in \C$. 
\end{theorem}
Denote by $[V]$ the average of $V$ over one periodicity cell, namely
\begin{equation*}
	[V]=\frac{1}{Q}\sum_{n\in W}V(n).
\end{equation*}
Denote by ${\bf 0}$ the zero potential. 
\begin{theorem}\label{thm1}
	\cite{liu1}
	Let $d=2$.
	The Laurent polynomial $ \mathcal{P}_V(z,\lambda)$ (as a function of $z\in\C^2$) is irreducible for any $\lambda\in \C$ except for $\lambda=[V] $. Moreover, if $ \mathcal{P}_{V}(z,[V])$ is reducible,
	$ \mathcal{P}_{V}(z,[V])=\mathcal{P}_{\bf 0}(z,0)$.
\end{theorem}

\begin{remark}
 \begin{enumerate}
	\item 
	The statements in	Theorem \ref{thm1} are sharp. Namely, 
	for $d=2$ and any constant function $V$, $F_{[V]}(V)/\Z^2$ has exactly two irreducible components. 
 
	\item Theorems \ref{thm2} and \ref{thm1} imply that the Fermi variety $F_{\lambda} (V)/\Z^d$ is irreducible except possibly  for
	$d=2$ and $\lambda=[V]$. 
	Partial results on
the irreducibility  of Fermi  and Bloch varieties, particularly 	in dimensions two and three,   were obtained in earlier  works ~\cite{GKTBook,ktcmh90,bktcm91,bat1,batcmh92,ls,battig1988toroidal,shva}.
	We should mention that 	reducible Fermi and Bloch  varieties  are known to occur for  periodic graph operators,  e.g.,  ~\cite{shi1,fls}. 
	\item Fillman, Liu and Matos \cite{flm21} have  applied   the approach introduced in ~\cite{liu1} to study   discrete periodic Schr\"odinger operators on general lattices.
	 \end{enumerate}
\end{remark}


 \begin{lemma}\label{le1}
 	Let $d\geq 1$.
 	Assume that 
 	${F}_{\lambda_0} (V)={F}_{\lambda_0} (Y)$ for some $\lambda_0\in\C$. Then 
 	$\mathcal{P}_V(z,\lambda_0) =\mathcal{P}_{Y}(z,\lambda_0)$.
 \end{lemma}
 \begin{proof}
 	When $d\geq 2$, Lemma \ref{le1} follows from Theorems \ref{thm2},  \ref{thm1} and \eqref{g37}. When $d=1$, by  \eqref{g37}, one has that 
 	\begin{equation*}
 	\mathcal{P}_V(z,\lambda)=(-1)^{q_1+1}z+(-1)^{q_1+1}z^{-1}+f_V(\lambda),
 	\end{equation*}
 	where $f_V(\lambda)$ does not depend on $z$. This implies Lemma \ref{le1} for $d=1$.
 \end{proof}

 \begin{remark}\label{re3}
 	Equality \eqref{g60} and Lemma \ref{le1} imply
 	\begin{enumerate}
 		\item $V$ and $Y$ are Floquet isospectral if and only if $\lambda^m_V(k)=\lambda^m_Y(k)$ for any $k\in\R^d$, $m=1,2,\cdots,Q$;
 			\item $V$ and $Y$ are Floquet isospectral  if and  only if  $\mathcal{P}_V(z,\lambda) =\mathcal{P}_{Y}(z,\lambda)$ for any  $\lambda\in\C$;
 		\item $V$ and $Y$ are  Fermi isospectral if and only if   $\mathcal{P}_V(z,\lambda_0) =\mathcal{P}_{Y}(z,\lambda_0)$ for some $\lambda_0\in\C$;
 			\item $V$ and $Y$ are  Fermi isospectral  if and only if there exists some $\lambda_0\in\C$ such that  for any $k\in\R^d$,  
 		\begin{equation*}
 		\prod_{m=1}^Q(\lambda^m_V(k)-\lambda_0)=	\prod_{m=1}^Q(\lambda^m_Y(k)-\lambda_0).
 		\end{equation*}
 	\end{enumerate}
 \end{remark}

Define the discrete Fourier transform $\hat{V}(l) $ for $l\in {W}$ by 
\begin{equation*}
\hat{V}(l) =\frac{1}{{Q}}\sum_{ n\in {W} } V(n) \exp\left\{-2\pi i \left(\sum_{j=1}^d \frac{l_j n_j}{q_j} \right)\right\}.
\end{equation*}
For convenience, we extend $\hat{V}(l)$ to $ \Z^d$    periodically, namely, for any $l\equiv m\mod \Gamma$,
\begin{equation*}
\hat{V}(l)=\hat{V}(m).
\end{equation*}

Define
\begin{equation}\label{gtm}
\tilde{\mathcal{D}}_V(z)= \tilde{\mathcal{D}}_V(z_1,z_2,\cdots,z_d)= \mathcal{D}_V(z_1^{q_1},z_2^{q_2},\cdots,z_d^{q_d}),
\end{equation}
and 
\begin{equation}\label{gtp}
\tilde{\mathcal{P}}_V(z,\lambda)=\det( \tilde{\mathcal{D}}_V(z,\lambda)-\lambda I)= \mathcal{P}_V(z_1^{q_1},z_2^{q_2},\cdots,z_d^{q_d},\lambda).
\end{equation}
Let $$\rho^j_{n_j}=e^{2\pi  \frac{n_j}{q_j} i},$$
where $0\leq n_j \leq q_j-1$, $j=1,2,\cdots,d$.

	When $q_{1}$ and $q_2$ are coprime, one has that
\begin{equation}\label{dckey1}
\rho_{n_{1}}^{1}- \rho_{n_{2}}^{2}\neq 0,
\end{equation}
provided that $ (n_{1},n_{2})\neq (0,0)\mod q_{1}\Z\oplus q_2\Z$.

A straightforward application of the  discrete Floquet transform (e.g.,  \cite{liu1,ksurvey}) leads to
\begin{lemma}\label{lesep} 
	Let $n=(n_1,n_2,\cdots,n_d) \in {W}$ and $n^\prime=(n_1^\prime,n_2^\prime,\cdots,n_d^\prime) \in {W}$. Then 
	$\tilde{\mathcal{D}}_V(z)$ is unitarily equivalent to 		
	$
	A+B_V,
	$
	where $A$ is a diagonal matrix with entries
	\begin{equation}\label{A}
	A(n;n^\prime)=\left(\sum_{j=1}^d \left(\rho^j_{n_j}z_j+\frac{1}{\rho^j_{n_j} z_j} \right)\right) \delta_{n,n^{\prime}}
	\end{equation}
	and \begin{equation}\label{gb}
	B_V(n;n^\prime)=\hat{V} \left(n_1-n_1^\prime,n_2-n_2^\prime,\cdots, n_d-n_d^\prime\right).
	\end{equation}
	In particular,
	\begin{equation*}
	\tilde{\mathcal{P}}_V(z, \lambda) =\det(A+B_V-\lambda I).
	\end{equation*}

\end{lemma}
\section{Technical preparations, I}\label{S3}

Assume $\sum_{j=1}^r d_j=d$ with $r\geq 2$ and $d_j\geq 1$, $j=1,2,\cdots, r$. For convenience, let $d_0=0$.
For any $l=(l_1,l_2,\cdots, l_d)\in\Z^d$, denote by 
\begin{equation}\label{g67}
\tilde{l}_j=(l_{d_{j-1}+1},l_{d_{j-1}+2},\cdots,l_{d_j}), j=1,2,\cdots,r.
\end{equation}

For any $n=(n_1,n_2,\cdots, n_d)\in\Z^d$, denote by 
\begin{equation}
\tilde{n}_j=(n_{d_{j-1}+1},n_{d_{j-1}+2},\cdots,n_{d_j}), j=1,2,\cdots,r.
\end{equation}
For $j=1,2,\cdots, r$, denote by 
\begin{equation*}
{W}_j= \{\tilde{n}_j:  0\leq n_m\leq q_m-1, m=d_{j-1}+1,d_{j-1}+2,\cdots, d_j\}.
\end{equation*}
\begin{lemma}\label{leSeparable1}
A potential 
	$V$ is $(d_1,d_2,\cdots,d_r)$ separable    if and only if     for any $l\in W$    with  at least two non-zero  
	$ \tilde{l}_{j},$ $j=1,2,\cdots, r$,  $
	\hat{V}(l)=0.
	$
\end{lemma}

\begin{proof}
	Suppose  $V$ is  $(d_1,d_2,\cdots,d_r)$ separable  on $\Z^d$. Then   
	there exist functions $V_j$ on $\Z^{d_j}$, $j=1,2,\cdots, r$, such that
	$$V(n)=\sum_{j=1}^rV_j(\tilde{n}_j).$$
A	direct computation shows that  for any $l$ with at least two non-zero  $\tilde{l}_j$, $j=1,2,\cdots, r$, 
	\begin{equation}\label{g38}
	\hat{V}(l)=0.
	\end{equation}
	Moreover, for any $\tilde{l}_j\in W_j$, 
	\begin{equation}\label{g39}
	\hat{V}(0,\cdots,0,  \tilde{l}_j,0,\cdots, 0)=\hat{V}_j(\tilde{l}_j)+\delta_{\tilde{l}_j}\sum_{m=1,m\neq j}^r[V_m],
	\end{equation}
	where for $\tilde{l}_j=(0,0,\cdots,0)$,  $\delta_{\tilde{l}_j}=1$, otherwise, $\delta_{\tilde{l}_j}=0$.
	By \eqref{g38}, 	we finish one side of the proof.  
Assume that $ V$ satisfies    $\hat{V}(l)=0$   for any $l=(\tilde{l}_1,\tilde{l}_2,\cdots,\tilde{l}_r)\in W$   with  at least two nonzero 
$  \tilde{l}_{j},$ $j=1,2,\cdots, r$.
	Let $V_j$ be such that its discrete Fourier transform $\hat{V}_j$ satisfies for any $\tilde{l}_j\in W_j\backslash\{0,0,\cdots,0\}$
	\begin{equation}\label{g56}
	\hat{V}_j(\tilde{l}_j)= 	\hat{V}(0,\cdots,0,  \tilde{l}_j,0,\cdots, 0)
	\end{equation}
	and 
	$[V_j]=\hat{V}_j((0,0,\cdots,0))=\frac{1}{r} [V]$. Let $\tilde{V}(n)=\sum_{j=1}^r V_j(\tilde{n}_j),n\in \Z^d$. By \eqref{g39}, one has that
	for any $\tilde{l}_j\in W_j$, 
	\begin{equation}\label{g40}
	\hat{V}(0,\cdots,0,  \tilde{l}_j,0,\cdots, 0)= \hat{\tilde{V}}(0,\cdots,0,  \tilde{l}_j,0,\cdots, 0).
	\end{equation}
	Since $ \hat{V}(l)=\hat{\tilde{V}}(l)=0$ for any $l$ with  at least two  non-zero 
	$  \tilde{l}_{j},$ $j=1,2,\cdots, r$, we have 
	\begin{equation*}
	\hat{V}(l)=\hat{\tilde{V}}(l),l\in W.
	\end{equation*}
	This implies $V=\tilde{V}$ and hence $V$ is  $(d_1,d_2,\cdots,d_r)$ separable.
	
\end{proof}

Recall that $\rho^j_{l} =e^{2\pi i \frac{l}{q_j}}$.
	\begin{lemma} \label{key2}
		Assume that $q_1$, $q_2$ and $q_3$ are relatively prime. Assume that
		\begin{equation}\label{g23}
\det \begin{pmatrix}
1 & 1 & 1  \\
\rho^1_{l_1} & \rho^2_{l_2} & \rho^3_{l_3}  \\
\rho^1_{l_1'} & \rho^2_{l_2'} & \rho^3_{l_3'} \\
\end{pmatrix}=0,
		\end{equation}
		where $l_j,l_j^\prime \in \{0,1,\cdots, q_j-1\}, j=1,2,3$. Then  $(l_1,l_2,l_3)$ and  $(l'_1,l'_2,l'_3)$ must  fall into one of the following cases:
		\begin{equation*}
		l_1=l_2=l_3=0;
		\end{equation*}
		\begin{equation*}
		l_1'=l_2'=l_3'=0;
		\end{equation*}
			\begin{equation*}
		l_1=l_1', l_2=l_2', l_3=l_3';
		\end{equation*}
		\begin{equation*}
		l_1=l_1^\prime=0, l_2=l_2^\prime=0;
		\end{equation*}
			\begin{equation*}
		l_1=l_1^\prime=0, l_3=l_3^\prime=0;
		\end{equation*}
			\begin{equation*}
		l_2=l_2^\prime=0, l_3=l_3^\prime=0.
		\end{equation*}
	\end{lemma}

\begin{proof}
	By row reduction and \eqref{g23}, one has
		\begin{equation*} 
	\det \begin{pmatrix}
	1 & 1 & 1  \\
	\rho^1_{l_1}- \rho^3_{l_3}& \rho^2_{l_2} -\rho^3_{l_3} & 0 \\
	\rho^1_{l_1'} -\rho^3_{l_3'} & \rho^2_{l_2'}-\rho^3_{l_3'}  & 0\\
	\end{pmatrix}=0,
	\end{equation*}
	and hence
	\begin{equation}\label{g24}
(	\rho^1_{l_1}- \rho^3_{l_3})(\rho^2_{l_2'}-\rho^3_{l_3'} )=( \rho^2_{l_2} -\rho^3_{l_3} ) (	\rho^1_{l_1'} -\rho^3_{l_3'}).
	\end{equation}
Due to the equality $e^{2i x}-e^{2i y}=2ie^{i(x+y)} \sin (x-y)$  and \eqref{g24}, we have that 
		\begin{align}
 e^{ i\pi \left( \frac{l_1}{q_1}+\frac{l_3}{q_3}\right)} & \left(\sin \pi \left( \frac{l_1}{q_1}-\frac{l_3}{q_3}\right)\right) e^{ i\pi \left( \frac{l_2'}{q_2}+\frac{l_3'}{q_3}\right)} \left(\sin \pi \left( \frac{l_2'}{q_2}-\frac{l_3'}{q_3}\right)\right)\nonumber\\=&e^{ i\pi \left( \frac{l_1'}{q_1}+\frac{l_3'}{q_3}\right)} \left(\sin \pi \left( \frac{l_1'}{q_1}-\frac{l_3'}{q_3}\right)\right) e^{ i\pi \left( \frac{l_2}{q_2}+\frac{l_3}{q_3}\right)} \left(\sin \pi \left( \frac{l_2}{q_2}-\frac{l_3}{q_3}\right)\right). \label{g25}
	\end{align}
	
	We prove Lemma \ref{key2} by cases.
	
	{\bf Case 1}: $l_2=l_3=0$ 
	
	In this case, by \eqref{g25}, one must have either
	\begin{equation*}
\sin \pi \left( \frac{l_1}{q_1}-\frac{l_3}{q_3}\right)=0
	\end{equation*}
	or 
		\begin{equation*}
	\sin \pi \left( \frac{l_2'}{q_2}-\frac{l_3'}{q_3}\right)=0.
	\end{equation*}
	By the fact that $q_1$, $q_2$ and $q_3$ are relatively prime and \eqref{dckey1}, one has either
	$l_1=l_3=0$ or $l_2'=l_3'=0$.
	
		{\bf Case 2}: $l_1'=l_3'=0$
		
		In this case, by the similar argument in Case 1, one also has either 	$l_1=l_3=0$ or $l_2'=l_3'=0$.
		
			{\bf Case 3}: $(l_2,l_3)\neq (0,0)$ and $(l_1',l_3')\neq (0,0)$
			
			In this case, the right side of \eqref{g25} is nonzero and then we have
			\begin{equation}\label{g26}
		e^{i\pi \left( \left(\frac{l_1'}{q_1}+\frac{l_3'}{q_3}+ \frac{l_2}{q_2}+\frac{l_3}{q_3}\right)-\left( \frac{l_1}{q_1}+\frac{l_3}{q_3}+\frac{l_2'}{q_2}+\frac{l_3'}{q_3}\right)\right) }\in \R.
			\end{equation}
			This implies that 
			\begin{equation*}
		\frac{l_1'-l_1}{q_1} +\frac{l_2-l_2'}{q_2} =0\mod\Z.
			\end{equation*}
			Therefore by \eqref{dckey1}, we have
			\begin{equation}\label{g27}
		 l_1=l_1' \text{ and } l_2=l_2'.
			\end{equation}
			By \eqref{g23}, \eqref{g27} and row reduction, we have 
				\begin{equation}\label{g28}
			\det \begin{pmatrix}
			1 & 1 & 1  \\
			\rho^1_{l_1} & \rho^2_{l_2} & \rho^3_{l_3} \\
		0 & 0 & \rho^3_{l_3'}-\rho^3_{l_3} \\
			\end{pmatrix}=0,
			\end{equation}
			and hence
			\begin{equation*}
			(\rho^3_{l_3'}-\rho^3_{l_3} )(	\rho^1_{l_1}-\rho^2_{l_2} )=0.
			\end{equation*}
			This implies either $l_3=l_3'$ or $l_1=l_2=0$.
			
\end{proof}
\section{Technical preparations, II}\label{S4}

For $a=(a_1,a_2,\cdots,a_d) \in \Z^d$ and $z^a=z_1^{a_1}z_2^{a_2}\cdots z_d^{a_d}$, denote by $|a|=a_1+a_2+\cdots+a_d$ the degree of $z^a$. 
\begin{theorem}\label{key1}
	Assume that  $V$ and $Y$ are Fermi isospectral. Then
	\begin{equation}\label{g21}
	[V]=[Y]
	\end{equation}
	and  for all possible $z\in\C^d$,
	\begin{equation}\label{g55}
	\sum_{ n\in W, n^\prime \in W}
	\frac{|\hat{V}(n-n^\prime)|^2}{(\sum_{j=1}^d \rho^j_{n_j}z_j) (\sum_{j=1}^d \rho^j_{n^\prime_j}z_j) } \equiv \sum_{ n\in W, n^\prime \in W}
	\frac{|\hat{Y}(n-n^\prime)|^2}{(\sum_{j=1}^d \rho^j_{n_j}z_j) (\sum_{j=1}^d \rho^j_{n^\prime_j}z_j) }.
	\end{equation}
\end{theorem}
\begin{proof}
	Assume that $ F_{\lambda_0} (V)=F_{\lambda_0} (Y)$ for some $\lambda_0\in\C$. By Lemma \ref{le1}, one has
	\begin{equation}\label{equ1}
	\mathcal{P}_V(z,\lambda_0)=\mathcal{P}_{Y}(z,\lambda_0),
	\end{equation}
	and	 hence
	\begin{equation}\label{equ28}
	\tilde{\mathcal{P}}_V(z,\lambda_0)=\tilde{\mathcal{P}}_{Y}(z,\lambda_0).
	\end{equation}
	For any $n=(n_1,n_2,\cdots,n_d)\in W$, let
	\begin{equation}
	t_{n}(z,x)=x+\sum_{j=1}^d \left(\rho^j_{n_j}z_j+\frac{1}{\rho^j_{n_j}z_j} \right).
	\end{equation}
	Clearly, $t_{n}(z,[V]-\lambda)$ is the $n$-th diagonal entry of the matrix $A+B_V$.
	By Lemma \ref{lesep}, 
	direct calculations imply that 
	\begin{align}
	\tilde{\mathcal{P}}_V(z,\lambda)=&\det (A+B_V-\lambda I)\nonumber\\
	=&
	\prod_{n\in W}t_{n}(z,[V]-\lambda)-\frac{1}{2} \sum_{ n\in W, n^\prime \in W\atop{n\neq n^\prime} }
	\frac{\prod_{n\in W}t_{n}(z,[V]-\lambda) }{t_{n^\prime}(z,[V]-\lambda)t_{n}(z,[V]-\lambda)} |B_V(n;n^\prime)|^2\nonumber\\
	&+\text{ products of at most } q_1q_2\cdots q_d-3 \text{ terms of } t_{n}(z,[V]-\lambda).\label{gtmp}
	\end{align}
	Let 
	\begin{equation*}
	h(z )=\prod_{n\in W}\left(\sum_{j=1}^d \rho^j_{n_j} z_j\right),
	\end{equation*}
	and 
	$h^1_V(z,\lambda)$ be all terms in $\tilde{\mathcal{P}}_V(z,\lambda)$ consisting of $z^a$ with $|a|=q_1q_2\cdots q_d-1$.

	By \eqref{gtmp}, one has  that	$h^1_V(z,\lambda)$  must come from $\prod_{n\in W}t_{n}(z,[V]-\lambda)$ and hence  
	\begin{equation}\label{equ2}
	h^1_V(z,\lambda)=( [V]-\lambda)\left(\sum_{n\in W} \frac{h(z)}{ \sum_{j=1}^d \rho^j_{n_j}z_j}\right).
	\end{equation}
	
		By \eqref{equ28} and \eqref{equ2}, one has
	\begin{equation}\label{g1}
	\sum_{n\in W} \frac{[V]-\lambda_0}{ \sum_{j=1}^d \rho^j_{n_j}z_j}\equiv \sum_{n\in W} \frac{[Y]-\lambda_0}{ \sum_{j=1}^d \rho^j_{n_j}z_j}.
	\end{equation}

	Set $z_3=z_4=\cdots=z_{d}=0$. Letting $z_{1}\to -z_2$ and
	by \eqref{g1} and \eqref{dckey1},
	one has
	\begin{equation}\label{g2}
	[V]=[Y].
	\end{equation}
	
	Let  ${h}_V^2(z,\lambda)$ ($\tilde{h}_V^2(z,\lambda)$) be  all terms in  $\tilde{\mathcal{P}}_V(z,\lambda)$ ($\prod_{n\in W}t_{n}(z,[V]-\lambda)$)
	consisting of $z^a$ with $|a|=q_1q_2\cdots q_d-2$.
	
	By \eqref{g2}, one has that 
	$\tilde{h}_V^2(z,\lambda_0)\equiv\tilde{h}_Y^2(z,\lambda_0)$. Therefore, by \eqref{equ28}, we have that 
	\begin{equation}\label{g22}
	h^2_V(z,\lambda_0) -\tilde{h}_V^2(z,\lambda_0)=	h^2_Y(z,\lambda_0) -\tilde{h}_Y^2(z,\lambda_0).
	\end{equation}
	By \eqref{gtmp}, one has that
	\begin{equation}\label{equ3}
	h^2_V(z,\lambda)-	\tilde{h}^2_V(z,\lambda)=\frac{h(z)}{2} \sum_{ n\in W, n^\prime \in W\atop{n\neq n^\prime} }
	\frac{-|B_V(n;n^\prime)|^2}{(\sum_{j=1}^d \rho^j_{n_j}z_j) (\sum_{j=1}^d \rho^j_{n^\prime_j}z_j) }.
	\end{equation}

	By \eqref{g22} and \eqref{equ3}, one has
	\begin{equation}\label{g3}
	\sum_{ n\in W, n^\prime \in W\atop{n\neq n^\prime} }
	\frac{|B_V(n;n^\prime)|^2}{(\sum_{j=1}^d \rho^j_{n_j}z_j) (\sum_{j=1}^d \rho^j_{n^\prime_j}z_j) } \equiv \sum_{ n\in W, n^\prime \in W\atop{n\neq n^\prime} }
	\frac{|B_Y(n;n^\prime)|^2}{(\sum_{j=1}^d \rho^j_{n_j}z_j) (\sum_{j=1}^d \rho^j_{n^\prime_j}z_j) }  ,
	\end{equation}
	and hence
	\begin{equation}\label{g4}
	\sum_{ n\in W, n^\prime \in W\atop{n\neq n^\prime} }
	\frac{|\hat{V}(n-n^\prime)|^2}{(\sum_{j=1}^d \rho^j_{n_j}z_j) (\sum_{j=1}^d \rho^j_{n^\prime_j}z_j) } \equiv \sum_{ n\in W, n^\prime \in W\atop{n\neq n^\prime} }
	\frac{|\hat{Y}(n-n^\prime)|^2}{(\sum_{j=1}^d \rho^j_{n_j}z_j) (\sum_{j=1}^d \rho^j_{n^\prime_j}z_j) } .
	\end{equation}
	By \eqref{g2} and \eqref{g4}, one has \eqref{g55}.
\end{proof}
\begin{theorem}\label{key4}
	Assume $V$ and $ Y$ are Fermi isospectral.   Let  $d_1$ be such that $2\leq d_1 \leq d$. Then we have that 
	\begin{align}
	\sum_{ \substack { l'\in W\\ l'_m=0,1\leq m\leq d_1}}| \hat{V}(l')|^2 
	=	\sum_{ \substack { l'\in W\\ l'_m=0,1\leq m\leq d_1}} |\hat{Y}(l')|^2, \label{g4111}
	\end{align}
	where $l'=(l'_1,l'_2,\cdots,l'_d)$.
\end{theorem}

\begin{proof}
	
		By \eqref{g55}, one has
	\begin{equation}\label{g291}
	\sum_{ \substack {n\in W, {l}' \in W}}
	\frac{|\hat{V}({l}')|^2}{(\sum_{j=1}^d\rho^j_{n_j}z_j) (\sum_{j=1}^d \rho^j_{n_j+{l}'_j}z_j) } \equiv 	\sum_{ \substack {n\in W, {l}' \in W}} \frac{|\hat{Y}({l}')|^2}{(\sum_{j=1}^d\rho^j_{n_j}z_j) (\sum_{j=1}^d \rho^j_{n_j+{l}'_j}z_j) }.
	\end{equation}
	Since $q_j$, $j=1,2,\cdots,d$ are relatively prime, one has that for any $(n_1,n_2,\cdots,n_{d_1})\neq (0,0,\cdots,0)\mod \Gamma$, 
	vectors 	$(1,1,\cdots, 1)$ and $(\rho_{n_1}^1,\rho_{n_2}^2,\cdots,\rho_{n_{d_1}}^{d_1})$ are linearly independent.
	By some basic facts in linear algebra, one has that  there exists $z_m^0$, $m=1,2,\cdots,d_1$ such that  for any 
	$(n_1,n_2,\cdots,n_{d_1})\neq (0,0,\cdots,0)\mod\Gamma$,
	\begin{equation}\label{g70}
	\sum_{m=1}^{d_1}\rho_{n_m}^mz_m^0\neq 0 
	\end{equation}
	and
	\begin{equation}\label{g71}
	\sum_{m=1}^{d_1}z_m^0=0.
	\end{equation}
	In \eqref{g291}, 	set $z_{d_1+1}=z_{d_1+2}=\cdots=z_d=0$.  Therefore, one has
		\begin{equation}\label{g29111}
	\sum_{ \substack {0\leq n_m\leq q_m-1\\ 1\leq m\leq d_1\\ {l}' \in W}}
	\frac{|\hat{V}({l}')|^2}{(\sum_{j=1}^{d_1}\rho^j_{n_j}z_j) (\sum_{j=1}^{d_1} \rho^j_{n_j+{l}'_j}z_j) } \equiv \sum_{ \substack {0\leq n_m\leq q_m-1\\ 1\leq m\leq d_1\\ {l}' \in W}} \frac{|\hat{Y}({l}')|^2}{(\sum_{j=1}^{d_1}\rho^j_{n_j}z_j) (\sum_{j=1}^{d_1} \rho^j_{n_j+{l}'_j}z_j) }.
	\end{equation}
	 Let $z_m\to z_m^0$, $m=1,2,\cdots, d_1$. 
	By \eqref{g70} and \eqref{g71},  leading terms (second order terms) in the sum of \eqref{g29111}  consist  of  $\frac{1}{(\sum_{j=1}^{d_1}\rho^j_{n_j}z_j) (\sum_{j=1}^{d_1} \rho^j_{n_j+{l}'_j}z_j)} $ with $n_1=n_2=\cdots=n_{d_1}=0$ and ${l}'_1={l}'_2=\cdots =l'_{d_1} =0$. Therefore, one has \footnote{ The readers can  also understand the    following formula without mentioning leading terms:  multiplying the equality \eqref{g29111} by $(z_1+z_2+\cdots +z_{d_1})^2$ and then letting $z_m\to z_m^0$, $m=1,2,\cdots,d_1$.  }
	\begin{equation}\label{g42}
	\sum_{ \substack { l'\in W\\ l'_m=0,1\leq m\leq d_1}}| \hat{V}(l')|^2 
=	\sum_{ \substack { l'\in W\\ l'_m=0,1\leq m\leq d_1}} |\hat{Y}(l')|^2.
	\end{equation}

\end{proof}

\begin{theorem}\label{key3}
	Let $d\geq 3$, and  $V$ and $ Y$  be  Fermi isospectral. Assume that  $l_i\in\{1,2,\cdots, q_i-1\}$, $i=1,2$.  Then  
	\begin{align}
	\sum_{ \substack {l'\in W\\  l'_1\in\{ l_1,q_1-l_1\}\\ l'_2\in\{l_2,q_2-l_2\}}}| \hat{V}(l')|^2 =	\sum_{ \substack {l'\in W\\ l'_1\in\{ l_1,q_1-l_1\}\\ l'_2\in\{l_2,q_2-l_2\}}}| \hat{Y}(l')|^2 , \label{g411}
	\end{align}
		where $l'=(l'_1,l'_2,\cdots,l'_d)$.
\end{theorem}
\begin{proof}

 Choose any $ l_3\in\{0,1,\cdots, q_3-1\}$.  By the assumption, one has that  $l_1 $ and $l_2$ are non-zero.
  By Lemma  \ref{key2}, for any $l'_m\in\{0,1,\cdots, q_m-1\}$, $m=1,2,3$  with $ (l_1',l_2',l_3')\neq (0,0,0)$  and $ (l_1',l_2',l_3')\neq (l_1,l_2,l_3)$, one has
\begin{equation}\label{g30}
\det \begin{pmatrix}
1 & 1 & 1  \\
\rho^1_{l_1} & \rho^2_{l_2} & \rho^3_{l_3}  \\
\rho^1_{l_1'} & \rho^2_{l_2'} & \rho^3_{l_3'} \\
\end{pmatrix}\neq 0.
\end{equation}

Let $z^0_1$, $z^0_2$ and $z_3^0$ be a nonzero solution of the linear system
\begin{equation}\label{g31}
z_1^0+z_2^0+z_3^0=0, \rho_{l_1}^1 z_1^0+\rho_{l_2}^2z_2^0+\rho_{l_3}^3z_3^0=0.
\end{equation}
By \eqref{g30},
 for any $l'_m\in\{0,1,\cdots, q_m-1\}$, $m=1,2,3$  with $ (l_1',l_2',l_3')\neq (0,0,0)$  and $ (l_1',l_2',l_3')\neq (l_1,l_2,l_3)$, one has
\begin{equation}\label{g32}
\rho_{l_1'}^1 z_1^0+\rho_{l_2'}^2z_2^0+\rho_{l_3'}^3z_3^0\neq 0.
\end{equation}
When $d\geq 4$, let $z_4=z_5=\cdots=z_d=0$. Then by \eqref{g291}, we have that
	\begin{equation}\label{g2911}
\sum_{ \substack {0\leq n_m\leq q_m-1\\m=1,2,3\\ {l}' \in W}}
\frac{|\hat{V}({l}')|^2}{(\sum_{j=1}^3\rho^j_{n_j}z_j) (\sum_{j=1}^3 \rho^j_{n_j+l'_j}z_j) } \equiv \sum_{ \substack {0\leq n_m\leq q_m-1\\m=1,2,3\\ l' \in W}} \frac{|\hat{Y}({l}')|^2}{(\sum_{j=1}^3\rho^j_{n_j}z_j) (\sum_{j=1}^3 \rho^j_{n_j+l'_j}z_j) }.
\end{equation}

Applying Theorem \ref{key4} with $d_1=3$, one has that  for  any fixed  $n_1, n_2$ and $n_3$,  
\begin{equation}\label{g76}
\sum_{ \substack { {l}' \in W\\ l_m'=0,m=1,2,3}}
\frac{|\hat{V}({l}')|^2}{(\sum_{j=1}^3\rho^j_{n_j}z_j) (\sum_{j=1}^3 \rho^j_{n_j+{l}'_j}z_j) }=\sum_{ \substack { {l}' \in W\\ l_m'=0,m=1,2,3}}
\frac{|\hat{Y}({l}')|^2}{(\sum_{j=1}^3\rho^j_{n_j}z_j) (\sum_{j=1}^3 \rho^j_{n_j+{l}'_j}z_j) }.
\end{equation}
By \eqref{g2911} and \eqref{g76}, one has
        	\begin{equation}\label{g291111}
        \sum_{ \substack {0\leq n_m\leq q_m-1\\m=1,2,3\\ l' \in W\\(l'_1,l'_2,l'_3)\neq(0,0,0)}}
        \frac{|\hat{V}({l}')|^2}{(\sum_{j=1}^3\rho^j_{n_j}z_j) (\sum_{j=1}^3 \rho^j_{n_j+l'_j}z_j) } \equiv \sum_{ \substack {0\leq n_m\leq q_m-1\\m=1,2,3\\ l' \in W\\(l'_1,l'_2,l'_3)\neq(0,0,0)}} \frac{|\hat{Y}({l}')|^2}{(\sum_{j=1}^3\rho^j_{n_j}z_j) (\sum_{j=1}^3 \rho^j_{n_j+l'_j}z_j) }.
        \end{equation}
In \eqref{g291111}, 	  let $z_m\to z_m^0$, $m=1,2,3$. 
By \eqref{g31} and \eqref{g32}, leading terms (second order terms) in the sum of \eqref{g291111}  consist of $\frac{1}{ (\sum_{j=1}^{3}\rho^j_{n_j}z_j) (\sum_{j=1}^3 \rho^j_{n_j+{l}'_j}z_j) }$ that satisfies  one of the following cases: 
\begin{equation}\label{g74}
n_1=n_2=n_3=0, {l}'_1=l_1, {l}'_2=l_2,{l}'_3=l_3,
\end{equation}
\begin{equation}\label{g75}
n_1=l_1, n_2=l_2, n_3=l_3, {l}'_1=q_1-l_1, {l}'_2=q_2-l_2, l_3'=\left\{
\begin{array}{@{}ll@{}}
0, & \text{if}\ l_3=0 \\
q_3-l_3, & \text{otherwise}
\end{array}\right .
\end{equation}
By \eqref{g291111}-\eqref{g75},
we have that 
	\begin{align}
\sum_{ \substack {l'\in W\\ l'_m=l_m\\m=1,2,3}}| \hat{V}(l')|^2 + \sum_{ \substack {l'\in W\\ l'_m=q_m-l_m\\m=1,2,3}} |\hat{V}(l')|^2=\sum_{ \substack {l'\in W\\ l'_m=l_m\\m=1,2,3}}| \hat{Y}(l')|^2 + \sum_{ \substack {l'\in W\\ l'_m=q_m-l_m\\m=1,2,3}} |\hat{Y}(l')|^2. \label{g4112}
\end{align}
Since \eqref{g4112} holds for arbitrary $l_3$, one has \eqref{g411}.
 
\end{proof}

\section{Proof of Theorem \ref{thmmain2}}\label{S5}
\begin{proof}
	By Lemma \ref{leSeparable1}, it suffices to show that 
	 for any $l\in W$   with  at least two of 	$ \tilde{l}_{j},$ $j=1,2,\cdots, r$  non-zero,   $
	\hat{V}(l)=0.$ Without loss of generality assume $\tilde{l}_1=(l_1,\cdots,l_{d_1})$ and $\tilde{l}_2=(l_{d_1+1},\cdots,l_{d_1+d_2})$ are non-zero.   Without loss of generality, we further assume that $l_1\neq 0$ and 
	$l_{d_1+1}\neq 0$. Since $Y$ is $(d_1,d_2,\cdots, d_r)$ separable, by Lemma \ref{leSeparable1}  one has that 
	that for any $l'\in W$  with $l'_1 \neq 0$ and $l'_{d_1+1} \neq 0$, 
	\begin{equation*}
\hat{Y}(l') =0
	\end{equation*}
	and hence
	\begin{equation}\label{g68}
\sum_{ \substack { l'\in W \\l'_1=l_1\\ l'_{d_1+1}=l_{d_1+1} }}| \hat{Y}(l')|^2+ \sum_{ \substack { l'\in W \\l'_1=q_1-l_1\\ l'_{d_1+1}=q_{d_1+1}-l_{d_1+1} }}  |\hat{Y}(l')|^2=0.
	\end{equation}

	Applying  Theorem 	\ref{key3}  with possible permutations
 and \eqref{g68}, we have 
\begin{equation}
\sum_{ \substack { l'\in W \\{l'}_1=l_1\\ l'_{d_1+1}=l_{d_1+1} }}| \hat{V}(l')|^2+ \sum_{ \substack { l'\in W \\{l'}_1=q_1-l_1\\ l'_{d_1+1}=q_{d_1+1}-l_{d_1+1} }}  |\hat{V}(l')|^2=0.
\end{equation}
Therefore, $\hat{V}(l)=0$ and 
we finish the proof.
\end{proof}
\section{Proof of Theorem  \ref{thmmain3} } \label{S6}
In order to prove Theorem \ref{thmmain3} , we only need to prove 
\begin{theorem}\label{thmmain31}
	Let  $d=d_1+d_2$ with $d_1\geq 2$ and $d_2\geq 1$. Assume that both $V$ and $Y$ are $(d_1,d_2)$ separable, namely, there exist $V_1, Y_1$ on $\Z^{d_1}$ and $V_2, Y_2$ on $\Z^{d_2}$ such that $V=V_1\oplus V_2$ and $Y=Y_1\oplus Y_2$.   Assume that $Y$ and $V$ are Fermi isospectral.   Then,  for any $j=1,2$, up to a constant, 
	$V_j$ and $Y_j$  are Floquet isospectral.
\end{theorem}

Let $d_1+d_2=d$ with $d_1\geq 2$ and $d_2\geq 1$.  Define
$$\Gamma_1=q_1\Z\oplus \cdots \oplus q_{d_1}\Z, \text{ and } \Gamma_2=q_{d_1+1}\Z \oplus \cdots \oplus q_{d_1+d_2}\Z.$$
Let $W_i$ be a fundamental domain  for $\Gamma_i$, $i=1,2$:
\begin{equation*}
W_1=\{(n_1,n_2,\cdots,n_{d_1})\in\Z^{d_1}: 0\leq n_j\leq q_{j}-1, j=1,2,\cdots, d_1\},
\end{equation*}
and
\begin{equation*}
W_2=\{(n_{d_1+1}, n_{d_1+2},\cdots,n_{d_1+d_2})\in\Z^{d_2}: 0\leq n_j\leq q_{j}-1, j=d_1+1,2,\cdots,d_1+ d_2\}.
\end{equation*}
For any $n=(n_1,n_2,\cdots, n_d)\in\Z^d$, denote by 
\begin{equation}\label{g63}
\tilde{n}_1=(n_1,n_2,\cdots,n_{d_1}), \tilde{n}_2=(n_{d_1+1},n_{d_1+2},\cdots,n_{d_1+d_2}).
\end{equation}

\begin{proof}[\bf Proof of Theorem \ref{thmmain31}]
	Recall \eqref{g21}.
	Without loss of generality, assume $[V_1]=[V_2]=[Y_1]=[Y_2]=0$ and thus  $[V]=[Y]=0$.
	By the assumption   that $V$ and $Y$ are Fermi isospectral, one has that there exists $\lambda_0\in\C$ such that for any $z\in( \C^{\star})^{d_1+d_2}$,
	\begin{equation*} 
	\mathcal{P}_V(z,\lambda_0)=\mathcal{P}_{Y}(z,\lambda_0).
	\end{equation*}
	Therefore,  	
	\begin{equation}\label{equ281}
	\tilde{\mathcal{P}}_V(z,\lambda_0)=\tilde{\mathcal{P}}_{Y}(z,\lambda_0).
	\end{equation}
	Denote by  $\hat{z}_1=(z_2,z_3,\cdots,z_{d_1})\in \C^{d_1-1}$. For any $\hat{z}_1 \in (\C^{\star})^{d_1-1}$ and $\lambda\in \C$,  let 
	$z_1=z_1(\lambda_0,\lambda,\hat{z}_1)$ be  the unique  solution of 
	\begin{equation}\label{equ32}
	\sum_{j=1}^{d_1}\left(z_j+\frac{1}{   z_j}\right) =\lambda_0-\lambda,
	\end{equation}
satisfying  $ |z_1|=|z_1(\lambda_0,\lambda,\hat{z}_1)|\to \infty$ as $| \sum_{j=2}^{d_1} z_j|\to \infty$ and $|z_j|\to\infty$, $j=2,\cdots, d_1$.
	
	Denote by  $\hat{W}_1=W_1\backslash \{(0,0,\cdots, 0)\}$.
	Since $q_1$, $q_2$, $\cdots, q_{d_1}$ are  relatively prime,  one has that 
	for any  $\tilde{n}_1=(n_1,n_2,\cdots,n_{d_1})\in \hat{W}_1$,  as  $| \sum_{j=2}^{d_1} z_j|\to \infty$ and $|z_2|, |z_3|,\cdots, |z_{d_1}|\to \infty$,
	\begin{equation}\label{equ29}
	\sum_{j=1}^{d_1}	\left(\rho_{n_j}^j z_j+\frac{1}{ \rho_{n_j}^j z_j} \right) =O(1)+ \sum_{j=2}^{d_1}( -\rho_{n_1}^1+ \rho_{n_j}^j) z_j,
	\end{equation}
	and 
	\begin{equation}\label{ge1}
	\sum_{j=2}^{d_1}( -\rho_{n_1}^1+ \rho_{n_j}^j) z_j \text{ is not identically zero}.
	\end{equation}
	Rewrite \eqref{equ32} as,  for $(n_1,n_2,\cdots, n_{d_1})=(0,0,\cdots,0)$,
	\begin{equation}\label{equ33}
	\sum_{j=1}^{d_1}	\left(\rho_{n_j}^j z_j+\frac{1}{ \rho_{n_j}^j z_j} \right) =\lambda_0-\lambda. 
	\end{equation}
	Let \begin{equation}\label{geh}
	h(z_2,z_3,\cdots,z_{d_1})= \left(\prod _{n\in \hat{W}_1}\left(\sum_{j=2}^{d_1}( -\rho_{n_1}^1+ \rho_{n_j}^j) z_j\right) \right)^{|W_2|},
	\end{equation}
	where $|W_2|$ is the cardinality  of $W_2$ ($|W_2|=\prod_{j={d_1+1}}^{d_1+d_2} q_j$).
	Clearly, $h$ is a homogeneous polynomial of $z_2,z_3,\cdots,z_{d_1}$ with degree $Q-|W_2|$.
	Let  $$\tilde{z}_2=(z_{d_1+1},z_{d_1+2},\cdots, z_{d_1+d_2})\in\C^{d_2}.$$
	
	For any $n=(\tilde{n}_1,\tilde{n}_2) \in\Z^{d_1+d_2}$  and  $n'=(\tilde{n}_1',\tilde{n}_2')\in\Z^{d_1+d_2}$ when  $\tilde{n}_1=\tilde{n}_1'=(0,0,\cdots,0) $, 
by \eqref{gb} and \eqref{g56}, one has that     
	\begin{equation}\label{g64}
	B_V(n;n')=\hat{V}(0,\cdots,0, \tilde{n}_2-\tilde{n}'_2)=\hat{V}_2(\tilde{n}_2-\tilde{n}'_2)=B_{V_2} (\tilde{n}_2;\tilde{n}_2'),
	\end{equation}
	and 
	by \eqref{A} and \eqref{equ33}, one has that 
		\begin{equation}\label{A1}
	A(n;n^\prime)=\left(\lambda_0- \lambda+\sum_{j=d_1+1}^d \left(\rho^j_{n_j}z_j+\frac{1}{\rho^j_{n_j} z_j} \right)\right)  \delta_{\tilde{n}_2,\tilde{n}_2^{\prime}}.
	\end{equation}
	By  Lemma \ref{lesep} and \eqref{equ29}-\eqref{A1}, one has that  as $|z_2|,|z_3|, \cdots, |z_{d_1}|\to \infty$  and $| \sum_{j=2}^{d_1} z_j|\to \infty$,
	\begin{equation}\label{equ30}
	\tilde{ \mathcal{P}}_V(z_1,\hat{z}_1,  \tilde{z}_2,\lambda_0)  = \tilde{ \mathcal{P}}_{V_2}(\tilde{z}_2,\lambda)  h(\hat{z}_1)+O(1)\left(\sum_{j=2}^{d_1}|z_j|^{Q-|W_2|-1}\right).
	\end{equation} 
	By \eqref{equ281} and  \eqref{equ30}, we have that
	for any $\tilde{z}_2$ and  $\lambda$,
	\begin{equation*}
	\tilde{ \mathcal{P}}_{Y_2}(\tilde{z}_2,\lambda)= 	\tilde{ \mathcal{P}}_{V_2}(\tilde{z}_2,\lambda).
	\end{equation*}
	This implies 
	$ V_2$ and $Y_2$ are Floquet isospectral.
 
		Now we are in a position to show that  	$ V_1$ and $Y_1$ are Floquet isospectral.
	Interchanging $V_1$ and $V_2$,  and $Y_1$ and $Y_2$, we see that  
 $V_1$ and $Y_1$ are Floquet isospectral  when  $d_2\geq 2$. Thus,  we only need to consider  the   case $d_2 =1$. 
Let $R>0$ be large enough and  set $\Omega=\{z_d\in\C:|z_d|> R\}$. 

Fixing $z_d$, solve the algebraic equation 
$	\tilde{\mathcal{P}}_{V_2}(z_d, \lambda) =0$.  By Lemma \ref{lesep},     there exist      solutions $\lambda^l(z_d)$, $l=1,2,\cdots, q_d$ such that  $\lambda^l(z_d)$ is analytic in $\Omega$ and 
\begin{equation}\label{gj81}
\lambda^l(z_d)= e^{2\pi \frac{l}{q_d} i} z_d+O(1) \text{ as } |z_d|\to \infty.
\end{equation}
Since 	$ V_2$ and $Y_2$ are Floquet isospectral, we know that  $\lambda^l(z_d)$, $l=1,2,\cdots, q_d$, are also the solutions of  the algebraic equation 
$	\tilde{\mathcal{P}}_{Y_2}(z_d, \lambda) =0$.
Using the fact that both $V$ and $Y$ are $(d_1,d_2)$ separable, one has that 
\begin{equation}\label{gj82}
\mathcal{P}_V(z_1,z_2,\cdots,z_{d-1},z_d^{q_2},\lambda_0)= \prod_{l=1}^{q_d}  \mathcal{P}_{V_1} (z_1,z_2,\cdots,z_{d-1}, \lambda^l(z_d)-\lambda_0),
\end{equation}
and 
\begin{equation}\label{gj83}
\mathcal{P}_Y(z_1,z_2,\cdots,z_{d-1},z_d^{q_2},\lambda_0)= \prod_{l=1}^{q_d}  \mathcal{P}_{Y_1} (z_1,z_2,\cdots,z_{d-1},\lambda^l(z_d)-\lambda_0).
\end{equation}
By \eqref{equ281},  \eqref{gj82}, \eqref{gj83}, Theorems  \ref{thm2} and \ref{thm1}, and the unique factorization theorem, one has that there is some $l_0\in \{1,2,\cdots, q_d\}$ such that 
\begin{equation}\label{gj84}
   \mathcal{P}_{V_1} (z_1,z_2,\cdots,z_{d-1}, \lambda^1(z_d)-\lambda_0)=  \mathcal{P}_{Y_1} (z_1,z_2,\cdots,z_{d-1},\lambda^{l_0}(z_d)-\lambda_0).
\end{equation}

By \eqref{gj81}, \eqref{gj84} and the  fact that  $q_l$, $l=1,2,\cdots, d$, are relatively prime, one  concludes that 
$l_0=1$. By \eqref{gj84}, we have that   for any $(z_1,z_2,\cdots,z_{d-1})\in( \C^{\star})^{d_1}=( \C^{\star})^{d-1}$ and $z_d\in \Omega$, 
\begin{equation}\label{gj85}
\mathcal{P}_{V_1} (z_1,z_2,\cdots,z_{d-1}, \lambda^1(z_d)-\lambda_0)=  \mathcal{P}_{Y_1} (z_1,z_2,\cdots,z_{d-1},\lambda^{1}(z_d)-\lambda_0).
\end{equation}
This implies that  for any $(z_1,z_2,\cdots,z_{d-1})\in( \C^{\star})^{d-1}$ and $ \lambda\in\C$, 
\begin{equation}\label{gj86}
\mathcal{P}_{V_1} (z_1,z_2,\cdots,z_{d-1},\lambda)=  \mathcal{P}_{Y_1} (z_1,z_2,\cdots,z_{d-1},\lambda),
\end{equation}
and hence $V_1$ and $Y_1$ are Floquet isospectral. We complete the proof.

\end{proof}

\section{Proof of Theorems \ref{thmmain5} and \ref{thmmain} }\label{S7}
\begin{proof}[\bf Proof of Theorem  \ref{thmmain5}]
	By  Theorem \ref{mainthm}, one concludes  that there exist functions $V_1$ on $\Z^{\tilde{d}}$ and $V_2$ on $\Z^{d-\tilde{d}}$ such that $V=V_1\bigoplus V_2$. Moreover, $V_2$ and the zero  potential are Floquet isospectral.  
	 By Remark \ref{re3}, 
	$\tilde{\mathcal{P}}_{V_2}(z,\lambda)=\tilde{\mathcal{P}}_{\bf 0}(z,\lambda)$  for any $\lambda$ and$z\in(\C^{\star})^{d-\tilde{d}}$, where ${\bf 0}$ is the zero function on $\Z^{d-\tilde{d}}$.
	Recall that   $\tilde{\mathcal{P}}_{V_2}(z,\lambda)$ and $\tilde{\mathcal{P}}_{\bf 0}(z,\lambda)$ are  polynomials in $\lambda$  with degree $Q_1=q_{\tilde{d}+1}q_{\tilde{d}+2}\cdots q_{d}$.
	Applying   Lemma \ref{lesep} to 	calculate the coefficients of  $\lambda^{Q_1-1}$ and $\lambda^{Q_1-2}$ in  both $\tilde{\mathcal{P}}_{V_2}(z,\lambda)$ and $\tilde{\mathcal{P}}_{\bf 0}(z,\lambda)$ (similar to \eqref{g2} and \eqref{g3}), one has that 
	\begin{equation*}
	[V_2]=0,
	\end{equation*}
	and 
	\begin{equation}\label{equ25}
 \sum_{0\leq l_j\leq q_j-1\atop{\tilde{d}+1\leq j\leq d} } |\hat{V}_2(l_{\tilde{d}+1}, l_{\tilde{d}+2},\cdots,l_{{d}})|^2=0, 
	\end{equation}
	and  hence $V_2$ is the zero function.
	This implies that $V$  only depends on the first $\tilde{d}$ variables.
\end{proof}
\begin{proof}[\bf Proof of Theorem \ref{thmmain}]

 By Corollary \ref{coro2},  $V$ is completely separable. By Lemma \ref{leSeparable1},  for any $l$    that  at least two of 
 $ {l}_{j},$ $j=1,2,\cdots, d$  are non-zero,  $ \hat{V}(l)=0$.  Therefore, in order to prove Theorem \ref{thmmain}, it suffices
to show  $\hat{V}(l)=0$ when at least two of $l_j$, $j=1,2,\cdots,d$ are zero.  
	Without loss of generality, assume $l_1=l_2=0$. 
	Applying Theorem \ref{key4} with the zero potential $Y$ and $d_1=2$, one has
	 \begin{equation}\label{g29}
	\sum_{ \substack { l'\in W\\ l'_1=0,l'_2=0 }}| \hat{V}(l')|^2= 0 .
	 \end{equation}
This implies
 for any $l_j\in\{0, 1,2\cdots, q_j-1\}$, $j=3,4,\cdots,d$,
  $$\hat{V}(l_1,l_2,l_3,l_4,\cdots,l_d)=\hat{V}(0,0,l_3,l_4,\cdots,l_d)=0.$$
	We finish the proof.
\end{proof}

\section*{Acknowledgments}

W. Liu was 
supported by NSF  DMS-2000345 and DMS-2052572.
I wish to thank   Pavel Kurasov for some useful discussions which resulted in the formulation of  Theorem \ref{thmmain5}.
I   also  wish to  thank  the referee for  their careful  reading of the manuscript  which greatly improved the exposition. 


 \bibliographystyle{abbrv} 

\bibliography{absence}

\end{document}